\documentclass[journal]{IEEEtran}
\usepackage{cite}
\usepackage{xcolor}
\ifCLASSINFOpdf
  \usepackage{graphicx}
\else
\fi
\usepackage[cmex10]{amsmath}
\usepackage{amsthm}
\usepackage{physics}

\newtheorem{theorem}{Theorem}[section]
\newtheorem{lemma}[theorem]{Lemma}

\newtheorem*{remark}{Remark}

\begin{document}

\title{Identification and Stabilization of Critical Clusters in Inverter-Based Microgrids}



\author{Andrey Gorbunov ~\IEEEmembership{Student Member,~IEEE}, Petr Vorobev~\IEEEmembership{Member,~IEEE},\\and
Jimmy Chih-Hsien Peng ~\IEEEmembership{Member,~IEEE}%
}

\markboth{IEEE TRANSACTIONS ON POWER SYSTEMS}%
{Shell \MakeLowercase{\textit{et al.}}: Bare Demo of IEEEtran.cls for Journals}

\maketitle

\begin{abstract}
A new method for stability assessment of inverter-based microgrids is presented in this paper. It leverages the notion of critical clusters -- a localized group of inverters with parameters having the highest impact on the system stability. The spectrum of the weighted network admittance matrix is proposed to decompose a system into clusters and rank them based on their distances from the stability boundary. We show that each distinct eigenvalue of this matrix is associated with one cluster, and its eigenvectors reveal a set of inverters that participate most in the corresponding cluster. The least stable or unstable clusters correspond to higher values of respective eigenvalues of the weighted admittance matrix. We also establish an upper threshold for eigenvalues that determines the stability boundary of the entire system and demonstrate that this value depends only on the grid type (i.e. $R/X$ ratio of the network) and does not depend on the grid topology. Therefore, the proposed method provides the stability certificate based on this upper threshold and identifies the lines or inverter droop settings needed to be adjusted to restore or improve the stability.
\end{abstract}


\IEEEpeerreviewmaketitle


\section{Introduction}

\IEEEPARstart{P}{ower} electronics interfaced generation is becoming increasingly widespread in modern power systems, mainly due to the increase in the share of renewable energy sources. Historically driven by governmental policies, this process is becoming economically feasible with the reduction of prices for semiconductor devices, and it is now being forecasted that by $2030$ up to $80\%$ of all the generated electricity in the world will go through power electronics devices \cite{DepartmentofEnergy2011}. Thus, the so-called zero-inertia power systems - AC electric systems with no synchronous machines - have attracted a lot of attention from both research and industrial communities in recent years. Grid-forming inverters - the core elements of such systems are now thought to become the main components of future power systems \cite{matevosyan2019grid, DepartmentofEnergy2011}. Unsurprisingly, inverter-based microgrids have attracted considerable attention of the research community over the last years, with studies focused on different inverter control techniques, modelling approaches, security assessment, etc. The literature on the topic is vast, and the interested reader can refer to excellent reviews \cite{olivares2014trends,parhizi2015state,guerrero2012advanced}.

Droop-controlled inverters - inverters, mimicking the behavior of synchronous machines \cite{Pogaku2007}, are supposed to become the main building blocks for zero-inertia grids. Power sharing and voltage regulation capabilities, the possibility of parallel operation, and rather simple and universal operating principles make these inverters a promising solution for future power electronics dominated grids. However, already the early experiments \cite{barklund2008energy} revealed that small-signal stability could become a major issue for such systems. Despite the seeming similarity of droop-controlled inverters to synchronous machines, the secure regions of inverter control parameters can be the very narrow, and careful tuning of droop coefficients might be necessary to guarantee the stable operation of such systems. Moreover, further research also revealed that modelling approaches that were routinely employed for conventional power systems fail to perform well for microgrids \cite{mariani2014model,Vorobev2016,Nikolakakos2016}. 

While direct detailed modeling can be used to certify stability of inverter-based microgrids \cite{naderi2019interconnected}, the computation cost of such an approach grows with the size of the system, and it becomes impractical already for moderate-size grids with several inverters. There has been a lot of research in the past few years on reduced-order models that can allow simple and reliable stability assessment. In \cite{rasheduzzaman2015reduced}, a method is proposed to compute the coefficients of a reduced-order model numerically, which is sufficiently accurate to detect possible instabilities. However, such a model provides no insight into what parameters of the system influence it's stability the most. In \cite{Nikolakakos2016,nikolakakos2017reduced}, reduced-order models of different levels of complexity were analysed and the system states, critical for stability assessment, were determined. It was shown that contrary to conventional power systems, simple approaches based on time-scale separation do not perform well for microgrids. Moreover, in \cite{Nikolakakos2016}, it was discovered that instability in inverter-based systems is mostly driven by the so-called "critical clusters" -  groups of inverters with short interconnection lines, and the heuristic method was proposed to detect these critical clusters. Later, in \cite{Vorobev2016} and \cite{Vorobev2017}, the influence of critical clusters on system stability was confirmed by analytic models, although approximate ones.

This paper aims to provide the procedure for the identification of the critical clusters, starting from the dynamic model of a system. To this end, we proposed a novel method based on the analysis of the spectrum of a specially constructed weighted network admittance matrix. Under realistic assumption of the uniform $R/X$ ratio in the network, the system can be decomposed into several modal subsystems -- \emph{clusters}, each associated with one eigenvalue of the admittance matrix. Here, the \textit{criticality} of a cluster - it's proximity to stability boundary - is directly related to the magnitude of the corresponding eigenvalue. 

The main contributions of this paper can be summarized as follows:

\begin{enumerate}

\item A method for microgrid decomposition into a set of equivalent two-bus systems - \textit{clusters} - was proposed, based on the spectrum of the system weighted admittance matrix. The method also allows to sorting those \textit{clusters} in order of their proximity to the stability boundary.

\item We show that there exists a critical value of $\mu_{cr}$ corresponding to stability boundary, and this value depends only on the grid type ($R/X$ ration of lines), and not on the grid topology. Thus, the set of all eigenvalues $\mu_i$ for a given grid can be used as a set of metrics for stability assessment.

\item We establish the dependence of the eigenvalues $\mu_i$ on microgrid parameters. We show, in particular, that the addition of any line to the system or decreasing the impedance of any of the existing lines leads to stability margin reduction.

\end{enumerate}


\section{Problem Statement} \label{sec:problem_statement}
The concept of critical clusters and its applications will be illustrated using a four-inverter system consisting of two areas. Referring to Fig. \ref{fig:kundur_mmg}, each area has two inverter-based generating units. This configuration is similar to the two-area four-generator Kundur's test system in \cite{kundur1994power}. In essence, interconnected inverters exhibit the same type of power swing dynamics as synchronous generators -- groups of inverters oscillate against each other. Hence, the motivation for selecting this test system is to demonstrate that for inverter-based microgrids, local clusters/modes (associated either with Area I or II) are typically less damped than inter-area clusters/modes. This is contrary to conventional power systems, where the low-frequency inter-area modes are typically the least damped ones, whereas the local modes are usually rather well-damped \cite{rogers2012power}. 

\begin{figure}[ht]
    \centering
    \includegraphics[width=0.35\textwidth]{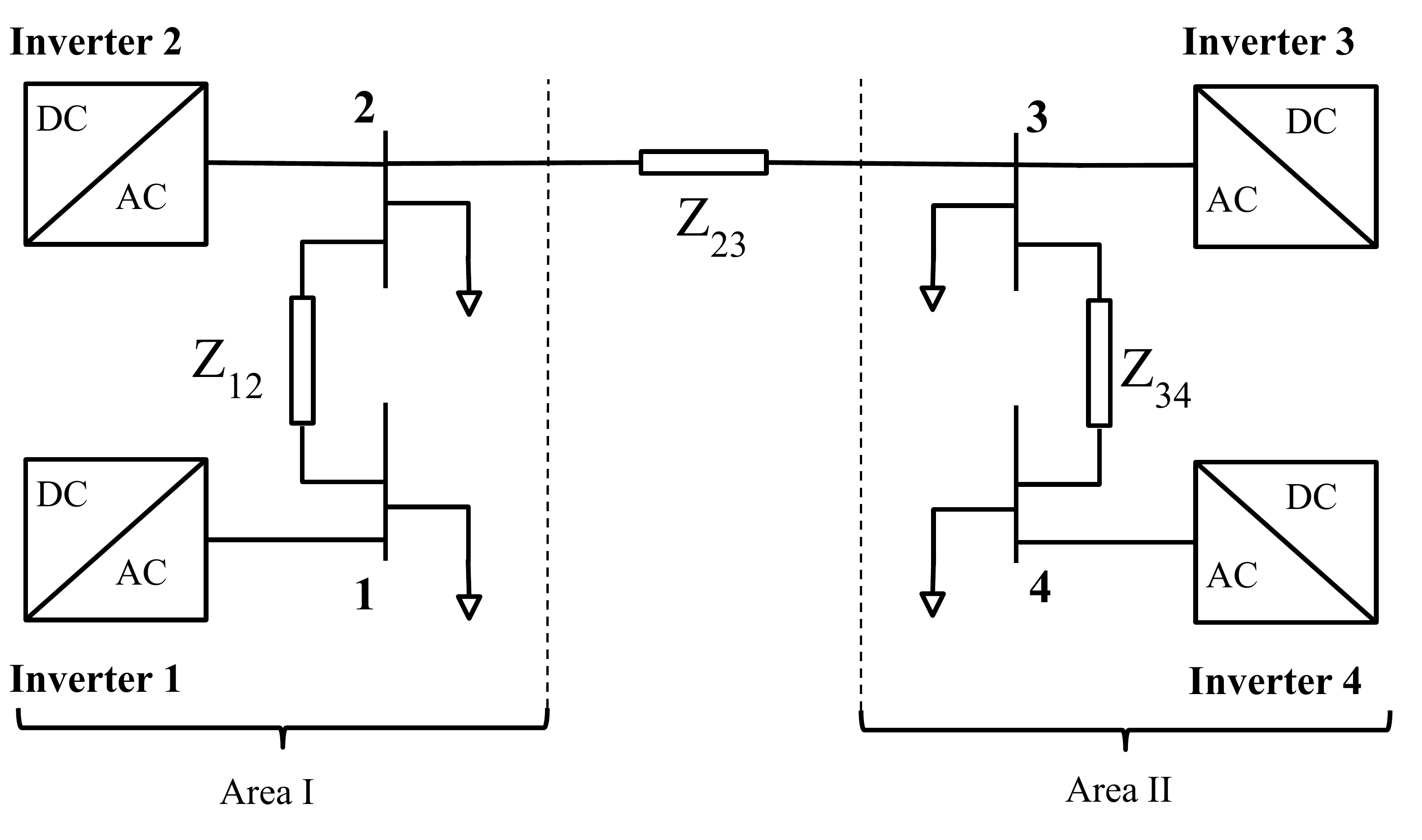}
    \caption{The two-area microgrid consisting of four inverters, operating in grid-forming mode (Network based on the $4$-bus model from \cite{kundur1994power}).}
    \label{fig:kundur_mmg}
\end{figure}

Here, we assume that all inverters are operating in the grid-forming mode, with droop control settings for both frequency and voltage. Details of this dynamic model are presented in Section \ref{sec:modeling}.

To demonstrate the property of clustering, the impedance of line $2-3$ is assumed to be significantly larger than those inside each area, i.e. $Z_{23}$ $>>$ $Z_{12}, Z_{34}$ (see Table \ref{tab:parameters}). Thereby, the network is distinctly separated into two areas, and the area that dominantly influences its stability can be subsequently determined. Fig. \ref{fig:kundur_eigenplot} illustrates the eigenvalues associated with the dominant modes of the system. It is observed that there are three dominant modes (pole-pairs), which correspond to oscillations between four inverters. Fig. \ref{fig:mode_shapes} shows the respective mode shapes (i.e. corresponding eigenvectors) -- the bigger the value of a mode shape element, the higher the magnitude of oscillations of the corresponding inverter compared to the others.
Here, the unstable mode shapes (red) are mainly associated with Area II. This is in agreement with the time-domain simulations presented in Fig. \ref{fig:step_response_load_change_m=3}, where all the inverters exhibit the growing amplitudes of frequency oscillations, but the amplitudes are significantly bigger for inverters $3$ and $4$. This observation motivates the investigation of the existence and identification of the critical clusters among inverter-based systems. 

\begin{figure}
    \centering
    \includegraphics[width=0.35\textwidth]{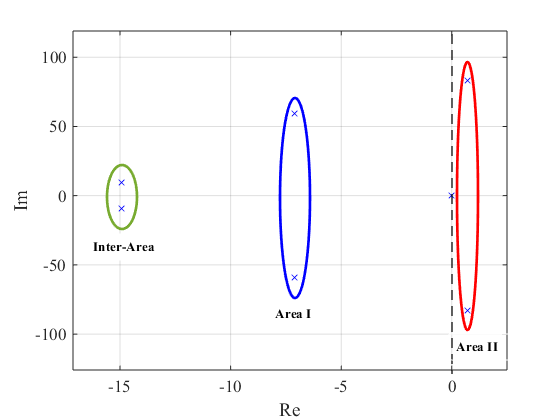}
    \caption{Low-frequency modes for the system of Fig.\ref{fig:kundur_mmg}, showing the local modes in Area I and Area II as well as the inter-area mode between both areas. Note that the pair of poles corresponding to unstable mode (i.e. critical mode) is in the right half-plane.}
    \label{fig:kundur_eigenplot}
\end{figure}

\begin{figure}
    \centering
    \includegraphics[width = 0.37\textwidth]{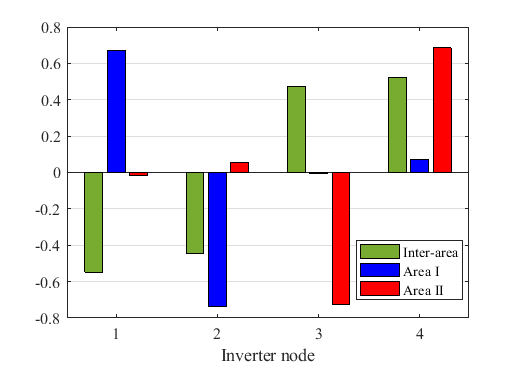}
    \caption{Mode shapes (eigenvectors) for the two-area system. Note that a higher absolute value indicates a larger oscillatory magnitude.}
    \label{fig:mode_shapes}
\end{figure}

\begin{figure}
    \centering
    \includegraphics[width=0.35\textwidth]{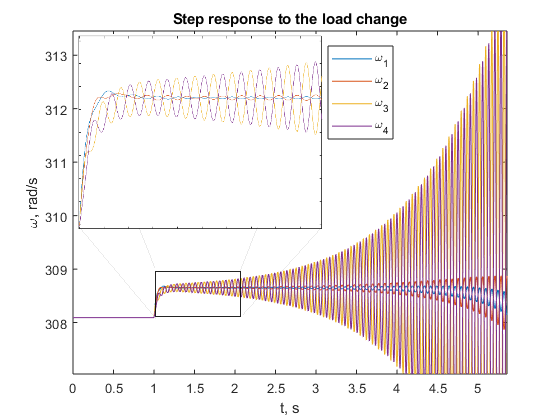}
    \caption{Dynamics of the system following a step load change ($10\%$). Both areas have unstable oscillations, but the amplitudes of oscillations are significantly higher in Area II.}
    \label{fig:step_response_load_change_m=3}
\end{figure}


\section{Modeling Approach} \label{sec:modeling}
In this section, we first recapitulate the dynamic model of inverter-based microgrids that is used for stability analysis. We adopted the $5$-th order electromagnetic model \cite{Vorobev2016}, which consists of states associated with dynamics of each inverter (angles $\theta_i$, frequency $\omega_i$, and voltage $V_i$) and each line (current components in \textit{dq} frame $I^{ij}_d$ and $I^{ij}_q$). Note that power, voltage, current, and impedances are given in per-unit. 
Denoting a set of lines as $\mathcal{E}$, a set of buses with inverters as $\mathcal{V}_{inv}$, and a set of all nodes as $\mathcal{V}$, the dynamic model constitutes of the following equations: 

\begin{subequations} \label{eq:dyn_nonlinear}
    \begin{align}
        &\dot{\theta_i} = \omega_i - \omega_0 \ ,\\
        &\tau \dot{\omega_i} = \omega_{set} - \omega_i - \omega_0 m_i P_i \ ,\\
        &\tau \dot{V_i} = V_{set} - V_i - n_i Q_i\ , \\
        &L^{ij} \dot{I^{ij}_d} = V_i \cos{\theta_i}-V_j \cos{\theta_j} - R^{ij} I^{ij}_d + \omega_0 L^{ij} I^{ij}_q \ ,\\
        &L^{ij} \dot{I^{ij}_q} = V_i \sin{\theta_i}-V_j \sin{\theta_j} - R^{ij} I^{ij}_q - \omega_0 L^{ij} I^{ij}_d.
    \end{align}
\end{subequations}
Where $\tau = \frac{1}{\omega_c}$ is the time constant of the power measurement low-pass filter \cite{Pogaku2007} with a cut-off frequency of $\omega_c$; $L^{ij}$ and $R^{ij}$ are the inductance and the resistance of the line connecting node $i$ and $j$, respectively. For brevity, the subscript of a variable represents the node index (e.g., $\theta_i, \ i \in \mathcal{V}_{inv}$), and the superscript refers to those of an edge (e.g., $I_q^{ij}, \ (ij) \in \mathcal{E}$).

Next, the linear approximation of \eqref{eq:dyn_nonlinear} is formulated. The state variables are derived from the stationary operating point, such that $\delta \theta_i = \theta_i^{\text{actual}} - \theta_i^0$, where $\theta_i^{\text{actual}}$ is the actual angle value, and $\theta_i^0$ is the operational value. Here, we use the fact that angular differences between inverters are typically small in inverter-based microgrids. Hence, the operating values can be assumed as $\theta_i^0 \approx 0 \ \forall i\in \mathcal{V}_{inv}$ and $V_i^0 \approx V_{set} = 1 \ p.u.$ \cite{mariani2014model,Vorobev2016}:

\begin{subequations} \label{eq:dyn_model}
    \begin{align}
        &\dot{\delta \theta_i} = \delta \omega_i \ , \\
        &\tau \dot{\delta \omega_i} = - \omega_i - \omega_0 m_i \delta P_i \  ,\\
        &\tau \dot{\delta V_i} = -\delta V_i - n_i \delta Q_i\ , \\
        &L^{ij} \dot{\delta I^{ij}_d} = \delta V_i - \delta V_j - R^{ij} \delta I^{ij}_d + \omega_0 L^{ij} \delta I^{ij}_q \ \label{eq:dyn_Id},\\
        &L^{ij} \dot{\delta I^{ij}_q} = \delta \theta_i - \delta \theta_j - R^{ij} \delta I^{ij}_q - \omega_0 L^{ij} \delta I^{ij}_d \label{eq:dyn_Iq}  \  .
    \end{align}
\end{subequations}
To simplify the notations, the state variables that are associated with deviations, e.g. $\delta \theta_i, \delta V_i$, etc., will be denoted merely as $\theta_i, V_i, \cdots$ from the hereafter. 

Given an arbitrary network, the number of lines can be greater than the number of buses, i.e. $ v \overset{\Delta}{=} |\mathcal{V}_{inv}| <  l \overset{\Delta}{=}|\mathcal{E}|$. For example, in a system with four buses ($v = 4$), it is possible to have a maximum of six lines ($l = \frac{v(v-1)}{2} = 6$). Accordingly, the number of equations in \eqref{eq:dyn_model} can become big even for moderate-size grids. In classic modeling, the use of nodal currents instead of line currents helps to reduce the number of variables. Subsequently, nodal currents are determined from nodal voltages through the admittance matrix (i.e. by a set of algebraic relations) $\pmb{I} = Y \pmb{V}$. However, such approach is possible in conventional power systems because line dynamics are usually neglected ($\dot{I_d}, \dot{I_q} \approx 0$). In inverter-based microgrids, explicit accounting for line dynamics is crucial for small-signal stability analysis, and derivative terms in equations \eqref{eq:dyn_Id} and \eqref{eq:dyn_Iq} can not be neglected \cite{Vorobev2016,Nikolakakos2016}. Nevertheless, a nodal representation of \eqref{eq:dyn_model} is still possible under certain modest assumptions \cite{tucci2016plug}, \cite{caliskan2012kron}, and will be discussed in the following subsection.
 
\subsection{Nodal representation of the network dynamics}
 
Suppose the incidence matrix of the power network is $\nabla^T$, which size is $|\mathcal{V}|\times |\mathcal{E}|$, such that $\nabla^T_{ij} = -1$ if the $j$th line leaves the $i$th node, or $\nabla^T_{ij} = 1$ if the $j$th line enters the $i$th node. The electromagnetic dynamics of the lines can then be represented by the following vector equations:
\begin{subequations} \label{eq:Ohms_law}
    \begin{align}
        &\pmb{\dot{\mathcal{I}}_d} = \omega_0 X^{-1}\nabla \pmb{V} - \omega_0 \mathcal{P} \pmb{\mathcal{I}_d} + \omega_0 \pmb{\mathcal{I}_q} \ ,\\
        &\pmb{\dot{\mathcal{I}}_q} = \omega_0 X^{-1}\nabla \pmb{\theta} - \omega_0 \mathcal{P}  \pmb{\mathcal{I}_q} - \omega_0 \pmb{\mathcal{I}_d}, \ .
    \end{align}
\end{subequations}
where $X = \text{diag}(\cdots,X^{ij},\cdots)$ is a matrix with line reactances on its diagonal; $\mathcal{P} =\text{ diag}(\cdots,\rho^{ij},\cdots)$ is a matrix with $R/X$ ratios on its diagonal; $\pmb{V}$ and $\pmb{\theta}$ are vectors consisting of voltages and phase angles at each bus; $\pmb{\mathcal{I}_d}$ and $\pmb{\mathcal{I}_q}$ are vectors of currents associated with each line. One should notice that the vectors $\pmb{V}, \pmb{\theta}$ and vectors $\pmb{\mathcal{I}_d}, \pmb{\mathcal{I}_q}$ have different dimensions: $|\mathcal{V}|$ and $|\mathcal{E}|$ accordingly. Let us multiply \eqref{eq:Ohms_law} by $\nabla^T$, and introduce the denotations $\pmb{I_d} = \nabla^T \pmb{\mathcal{I}_d}, \ \pmb{I_q} = \nabla^T \pmb{\mathcal{I}_q}$ for nodal currents, such that: 
\begin{subequations} \label{eq:EM_dynamics}
    \begin{align}
        &\dot{\pmb{I_d}} = \omega_0 \nabla^T X^{-1}\nabla \pmb{V} - \omega_0 \nabla^T \mathcal{P}  \pmb{\mathcal{I}_d} + \omega_0 \pmb{I_q} \ ,\\
        &\pmb{\dot{I_q}} = \omega_0 \nabla^T X^{-1}\nabla \pmb{\theta} - \omega_0 \nabla^T \mathcal{P}  \pmb{\mathcal{I}_q} - \omega_0 \pmb{I_d}, \ .
    \end{align}
\end{subequations}

 Equations \eqref{eq:EM_dynamics} are almost in the desired form with the vectors $\pmb{V}$, $\pmb{\theta}$ and $\pmb{I_d}, \pmb{I_q}$ being nodal, except for the presence of the term $\nabla^T \mathcal{P}  \pmb{\mathcal{I}}_d$. The remarkable (and practically important) case is when all the lines have the same R/X ratio, $\mathcal{P} = \rho \mathbf{1}$ (homogeneous $\rho$). In this case, $\nabla^T \mathcal{P}  \pmb{\mathcal{I}}_d = \nabla^T \rho \pmb{\mathcal{I}}_d = \rho \pmb{I_d}$ and the nodal representation is \cite{tucci2016plug}, \cite{caliskan2012kron}:

\begin{equation}\label{eq:EM_dynamics_homogen}
    \tau_0(\dot{\pmb{I_d}}+ j \dot{\pmb{I_d}}) = (1+\rho^2) B (\pmb{V} + j \pmb{\theta}) - (\rho + j) (\pmb{I_d}+ j \pmb{I_d}) \ ,
\end{equation}
where $B = -\Im{Y}$ is a susceptance matrix and $\tau_0 = \frac{1}{\omega_0}$ .

Apart from reducing the number of equations in \eqref{eq:Ohms_law} for a system with more lines than nodes, another advantage of \eqref{eq:EM_dynamics_homogen} is the fact, that Kron reduction procedure can be performed over it to additionally eliminate virtual nodes. Specifically, to eliminate a virtual node with zero net current injection $[I_d + j I_q]_i = 0$, one should apply Kron reduction to $B$ and use \eqref{eq:EM_dynamics_homogen} with the reduced $B$.

\subsection{The state-space model with homogeneous $\rho$} \label{subsec:same_rx}

To represent the dynamic model in the state-space form the real and reactive powers of an inverter are expressed in terms of inverter's voltage and current. Again, assuming $V_i^0 \approx  1 \ p.u.$ and $\theta_i^0 \approx 0$, the relationship between nodal power injections $\mathbf{P}, \mathbf{Q}$ and nodal currents $\pmb{I}_d, \pmb{I}_q$ can be expressed as follows:

\begin{equation}
    \begin{pmatrix}
        \mathbf{P} \\
        \mathbf{Q} 
    \end{pmatrix} = 
    \begin{bmatrix}
        \mathbf{1} & 0\\
        0 & -\mathbf{1}\\
    \end{bmatrix}
    \begin{pmatrix}
        \pmb{I}_d\\
        \pmb{I}_q
    \end{pmatrix} ,
\end{equation} 
Subsequently, using \eqref{eq:EM_dynamics_homogen}, the following state-space representation of the electromagnetic (EM) $5^{th}$ order model of a microgrid with the homogeneous $\rho$ is obtained:

\begin{equation} \label{eq:5th_model}
    \dot{
    \begin{pmatrix}
        \pmb{\theta}\\
        \tau\pmb{\omega}\\
        \tau\pmb{V}\\
        \tau_0\pmb{I}_d\\
        \tau_0\pmb{I}_q
    \end{pmatrix}} =
    \begin{bmatrix}
        0 & \mathbf{1} & 0 & 0 & 0 \\
        0 & -\mathbf{1} & 0 & -\omega_0 M & 0 \\
        0 & 0 & -\mathbf{1} & 0 & N\\
        0 & 0 & \mathcal{B} & -\rho \mathbf{1} & \mathbf{1} \\
        \mathcal{B} & 0 & 0 & -\mathbf{1} & -\rho \mathbf{1}
    \end{bmatrix}
    \begin{pmatrix}
        \pmb{\theta}\\
        \pmb{\omega}\\
        \pmb{V}\\
        \pmb{I}_d\\
        \pmb{I}_q
    \end{pmatrix} \ ,
\end{equation}
where $M = \text{diag}(m_1, \cdots, m_v), N = \text{diag}(n_1, \cdots, n_v)$ are diagonal matrices of droop gains, and $\mathcal{B} = (1+\rho^2) B$. 


\section{Eigenmodes decomposition theory}\label{sec:eigen}

The model in \eqref{eq:5th_model} can be written in a state-space form, $\dot{\pmb{x}}= A\pmb x$, such that stability analysis of the system \eqref{eq:dyn_model} can be evaluated by eigenvalue analysis of the matrix $A$. In this case, the state matrix in \eqref{eq:5th_model} has a specific structure -- a five by five block matrix where each sub-matrix is symmetric, and this property can be exploited further. We start by formulating the following Lemma:

\begin{lemma}
    The eigenvalues $\lambda$ of the state-space model \eqref{eq:5th_model} coincide with the eigenvalues of the following polynomial eigenvalue problem.
    \begin{equation} \label{eq:dyn_polynomial}
    \{f(\lambda) \tau M^{-1} + g(\lambda) (\mathcal{B} + \tau \lambda \mathcal{B} N M^{-1}) + \mathcal{B} N \mathcal{B}\} \pmb{\psi} = 0\ ,
    \end{equation}
    where $f(\lambda) = \lambda g^2(\lambda)[h^2(\lambda) + 1]$, $g(\lambda) = (1 + \tau \lambda)$, $h(\lambda) = (\rho + \tau_0 \lambda)$.
\end{lemma}

\begin{proof}
  For the purposes of this proof we represent the dynamic model \eqref{eq:5th_model} in the Laplace domain $[s\mathbf{1} - A] \pmb{x} = 0$. After some transformations one can exclude $\pmb{I_d}, \pmb{I_q}$ from \eqref{eq:5th_model} representing the dynamics in terms of $\pmb{V}$ and $\pmb{\theta}$ only such that:
  \begin{equation} \label{eq:dyn_thetaV}
    \begin{aligned}
    g(s)
    \begin{bmatrix}
        -h(s)\mathbf{1} & -\mathbf{1} \\
        -\mathbf{1} & h(s)\mathbf{1}
    \end{bmatrix}
    \begin{bmatrix}
        \tau s M^{-1} & 0 \\
        0 & N^{-1}
    \end{bmatrix}
    \begin{pmatrix}
        \pmb{\theta} \\
        \pmb{V}
    \end{pmatrix} = \\
    \begin{bmatrix}
        0 & \mathcal{B} \\
        \mathcal{B} & 0
    \end{bmatrix}
    \begin{pmatrix}
        \pmb{\theta} \\
        \pmb{V}
    \end{pmatrix}
\end{aligned} \ .
\end{equation}
Subsequently, $\pmb{V}$ can also be excluded from \eqref{eq:dyn_thetaV}, accordingly:
\begin{eqnarray} \label{eq:dyn_polynomial0}
    \{ -g(s)h(s)\tau sM^{-1} \nonumber \\ 
    - [g(s)N^{-1} + \mathcal{B}] \frac{N}{g(s)h(s)}[g(s)\tau sM^{-1} + \mathcal{B}] \}\pmb{\theta} = 0 \ .
\end{eqnarray}
And finally, the desired representation \eqref{eq:dyn_polynomial} is obtained by multiplying \eqref{eq:dyn_polynomial0} by $g(s)h(s)$:
\begin{equation} \label{eq:dyn_polynomial1}
    \{f(s) \tau M^{-1} + g(s) (\mathcal{B} + \tau s \mathcal{B} N M^{-1}) + \mathcal{B} N \mathcal{B}\} \pmb{\theta} = 0\ .
\end{equation}
The matrix in \eqref{eq:dyn_polynomial1} has zero determinant whenever $s = \lambda$, where $\lambda$ is the solution to the stated polynomial eigenvalue problem \eqref{eq:dyn_polynomial}.
\end{proof}

There is a remarkable case when the active and reactive power droop coefficients ratio is uniform across the system, i.e. in matrix terms $M = k N$, as it is suggested in \cite{Schiffer2016}. With this representation, the following theorem can be formulated, which is the base for all the manuscript results: 

\begin{theorem}
    The eigenvalues $\lambda$ of \eqref{eq:5th_model} under proportional droops $M = kN$ are connected with the eigenvalues $\{\mu_i, \ i = 1,\cdots, v\}$ of the weighted network susceptance matrix $C = M\mathcal{B}$ as follows:
    \begin{equation} \label{eq:two_bus_polynomial}
        \tau kf(\lambda) + g(\lambda)(k+\tau \lambda) \mu_i + \mu_i^2 = 0.
    \end{equation}
    
\end{theorem}

\begin{proof}
    
Taking into account the assumption $M  = kN$, the dynamic model \eqref{eq:dyn_polynomial} is equivalent to the following matrix polynomial of the matrix $C = M\mathcal{B}$:

\begin{equation} \label{eq:sym_polynomial}
    [\tau k f(\lambda) \mathbf{1} + g(\lambda) (k + \tau \lambda) C + C^2] \pmb{\psi} = 0 \ .
\end{equation}

It is noteworthy that eigenvectors of $C$ such that $C \pmb{\psi}_i = \mu_i \pmb{\psi}_i$ are also the eigenvectors of \eqref{eq:sym_polynomial} as $\pmb{\psi}_i$ is the eigenvector for other matrix terms in \eqref{eq:sym_polynomial}. Namely: $\mathbf{1}\pmb{\psi}_i = \pmb{\psi}_i$ and $C^2 \pmb{\psi}_i = \mu_i^2 \pmb{\psi}_i$. Using this fact, we obtain the connection between $\lambda$ and $\mu_i$ as in \eqref{eq:two_bus_polynomial}.

\end{proof}

 Thereby, the stability analysis of the whole system reduces to the analysis of polynomials \eqref{eq:two_bus_polynomial} for all values of $\mu_i$. Each polynomial in \eqref{eq:two_bus_polynomial} (i.e. polynomial with a specific $\mu_i$), in fact, corresponds to five distinct modes of the initial system and can be thought of as a representation of an equivalent two-bus system - a \emph{cluster}. Thus, the initial microgrid can be effectively split into separate clusters utilizing the spectrum of the weighted network susceptance matrix $C = M\mathcal{B}$. Fig. \ref{fig:kundur_decoupling} shows an example of such a split for the Kundur's system from Fig. \ref{fig:kundur_mmg}.  

Each cluster in Fig. \ref{fig:kundur_decoupling} represents a system, which is fully equivalent (in a small-signal sense) to an inverter versus an infinite bus system with the equivalent parameters: $m^{eq} = kn^{eq}$, $R^{eq} = \rho X^{eq}$ and $\frac{m^{eq}}{X^{eq}} = \mu_i$. Thus, the stability assessment of the initial grid is reduced to the stability assessment of a set of equivalent two-bus systems, for which any known rules and methods can be used. 

\begin{remark} \label{remark}
Each $\mu_i$ corresponds to five eigenvalues $\lambda$ that are the roots of \eqref{eq:two_bus_polynomial}. For example, five $\lambda$ associated with Area II, red dots in Fig. \ref{fig:kundur_eig_large_scale} (that is zoom-out version of Fig. \ref{fig:kundur_eigenplot}, where two unstable eigenvalues out of five are depicted), are the roots of \eqref{eq:two_bus_polynomial} with the highest $\mu = 2.06$. Hereby, the system \eqref{eq:5th_model} of $v$ inverters decouples into $v$ separate clusters each corresponding to one $\mu_i, \ i=1,\cdots,v$.
\end{remark}

It is important to note that all the roots of the equation \eqref{eq:two_bus_polynomial} have the negative real part if the corresponding $\mu_i$ is less than a certain value $\mu_{cr}$. If all $\mu_i < \mu_{cr}, \ i=1,\cdots,v$, then the system is stable, and vice versa -- if there some $\mu_i > \mu_{cr}, \ i=1,\cdots,u$, then the system has $u$ unstable modes. Hence, $\mu_i$ are the key values for stability assessment of the system, so that parameters of the system that affect the corresponding $\mu_i$ most also are the ones that should be changed to gain stability margin.

The critical value $\mu_{cr}$ for a given system depends on the network $\rho=R/X$ ratio (but not on the network topology) and the ratio of droop coefficients $k$. It can be found numerically for every pair of $\rho$ and $k$. Thus, one can say that a certain value of $\mu_{cr}$ corresponds to a class of networks. Plot of $\mu_{cr}$ for various $\rho$ and $k$ values can be found in our preliminary conference paper on the topic \cite{gorbunov2020identification}. With some degree of conservativeness, critical values $\mu_{cr}$ can be found using known results for two-bus systems. For example, using conservative stability from \cite[eq.~(17)]{vorobev2018towards}, the following analytic formulas for the lower bound of $\mu_{cr}$ can be written:

\begin{equation}\label{eq:mu_conservative}
    \mu_{cr} \geq \mu_{cr}^{lb} \overset{\Delta}{=} 
    \begin{cases}
    \frac{(\rho^2+1)^2}{4\rho} , \ k > \frac{1}{2} \rho(\rho^2 + 1)\\
    k \frac{\rho^2 + 1}{2\rho^2} , \ k \leq \frac{1}{2} \rho(\rho^2 + 1)\ .
    \end{cases}
\end{equation}

The next section discusses how the spectrum of $C$, $\mu_i$, depends on the system parameters. 

\begin{figure}
    \centering
    \includegraphics[width = 0.3\textwidth]{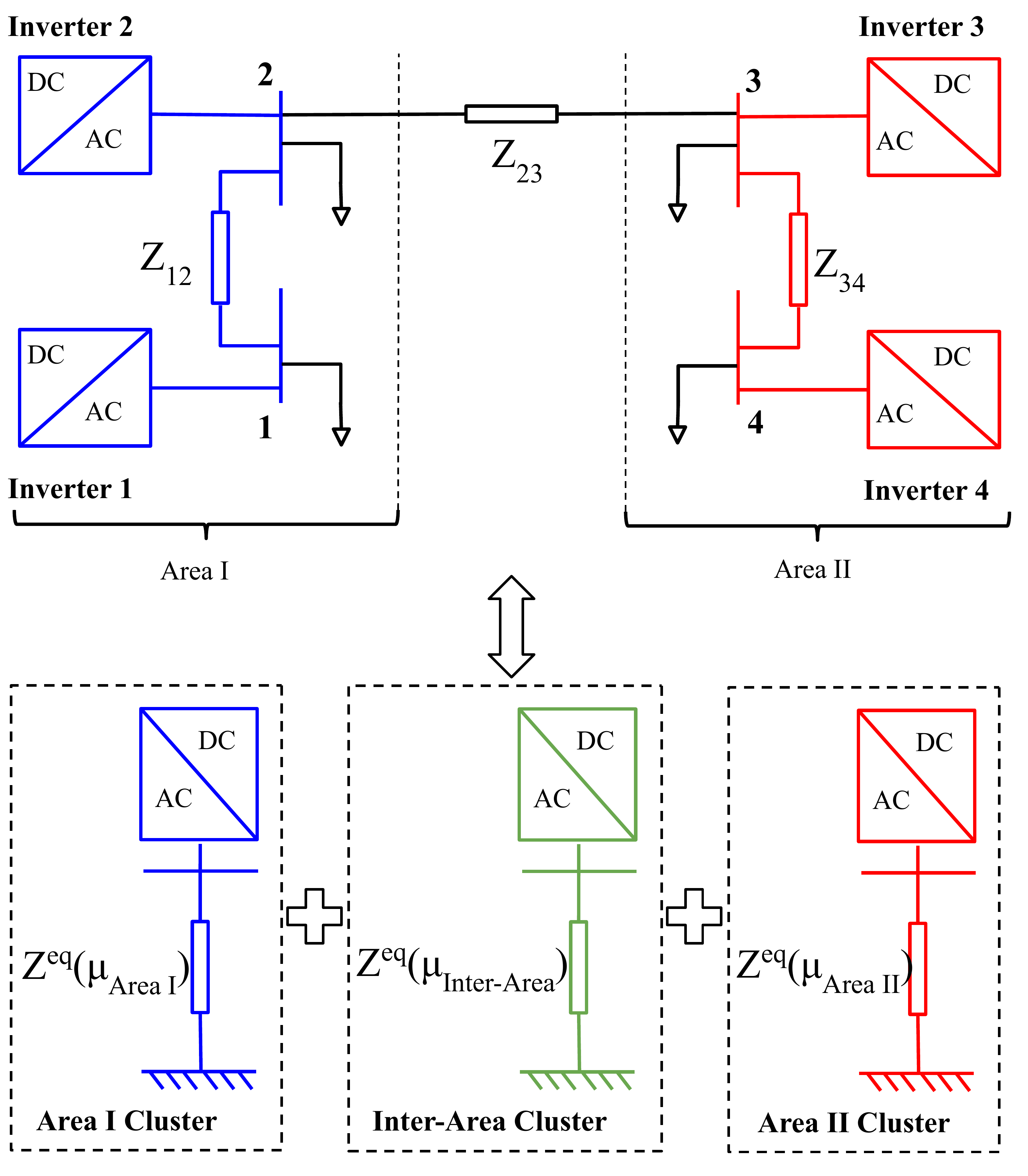}
    \caption{Illustration of decomposing a larger system into a set of clusters - equivalent two-bus systems.}
    \label{fig:kundur_decoupling}
\end{figure}

\begin{figure}
    \centering
    \includegraphics[width=0.4\textwidth]{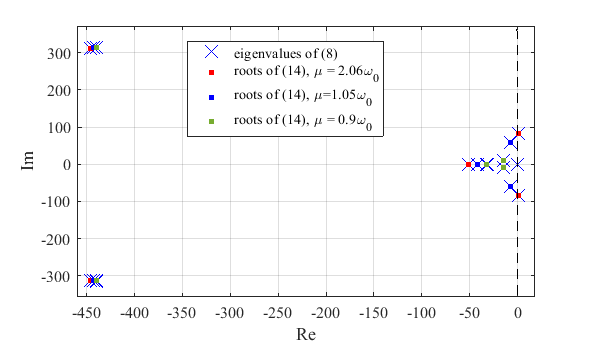}
    \caption{Pole-zero plot of the two-area system, including local modes, inter-area modes, and electromagnetic modes, with roots of \eqref{eq:two_bus_polynomial} associated with each cluster.}
    \label{fig:kundur_eig_large_scale}
\end{figure}

\subsection{Procedure for identifying critical clusters} \label{sec:identification_procedure}
 
As was established in the previous subsection, the stability of a microgrid can be evaluated by analysing the spectrum of its generalized Laplacian matrix $C=M\mathcal{B}$. Recall that \eqref{eq:two_bus_polynomial} only depends on the system $\rho=R/X$ ratio and the droop gain ratio $k$, but not on the network, line lengths, or values of droop gains. Therefore, one can say that $\mu_{cr}$ can be calculated for a class of networks simultaneously. 

Leveraging on these properties, the procedure for the stability assessment (Fig. \ref{fig:critical_cluster_flowchart}) is formulated as follows. First, determine $\mu_{cr}$ from \eqref{eq:two_bus_polynomial} provided the values of $\rho$ and $k$ 
 are given for the system. This can be done either directly (numerically) from \eqref{eq:two_bus_polynomial}, or using some conservative estimations, like \eqref{eq:mu_conservative}. Next, find the actual values of $\mu_i$ - the spectrum of the system matrix $C$. If all of the $\mu_i$ are less than $\mu_{cr}$, then the system is stable. If one or more of $\mu_i$ is greater than $\mu_{cr}$, then the system is unstable, and the corresponding eigenvectors $\pmb{\psi}_i$ give the structure for the corresponding unstable clusters. The high magnitudes of $\pmb{\psi}_i$ correspond to critical droop or line parameters, as it is shown in Section \ref{sec:sensitivity}.

The output of the proposed procedure will be a set of two-bus equivalent systems that are arranged in the order of decreasing $\mu$. Those corresponding to higher $\mu$ will be the ones with the least stability margin. Besides, Fig. \ref{fig:kundur_decoupling} illustrates the two-bus equivalent clusters on the example of the two-area system.  
The next section demonstrates the significant elements of the corresponding eigenvector $\pmb{\psi}$ suggest the parameters of the inverters and lines in these clusters that should be changed to stabilize the system or increase its stability margin.

\begin{figure}
    \centering
    \includegraphics[width=0.4\textwidth]{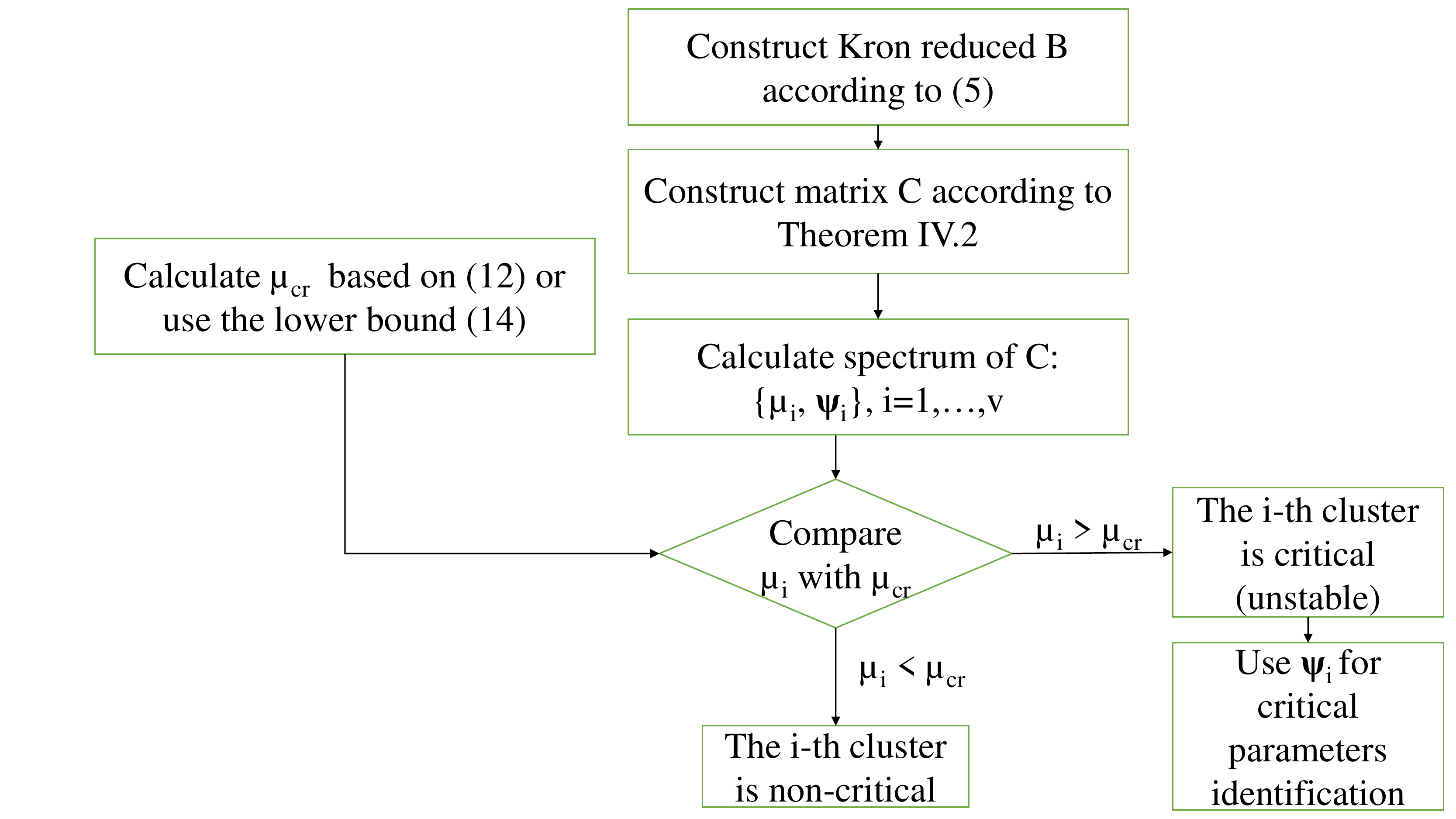}
    \caption{Flowchart to identify the critical clusters in an inverter-based system.}
    \label{fig:critical_cluster_flowchart}
\end{figure}


\section{Stability Margins and Sensitivity Αnalysis} \label{sec:sensitivity}

\subsection{Properties of the $C$ spectrum}

\begin{figure}
    \centering
    \includegraphics[width=0.3\textwidth]{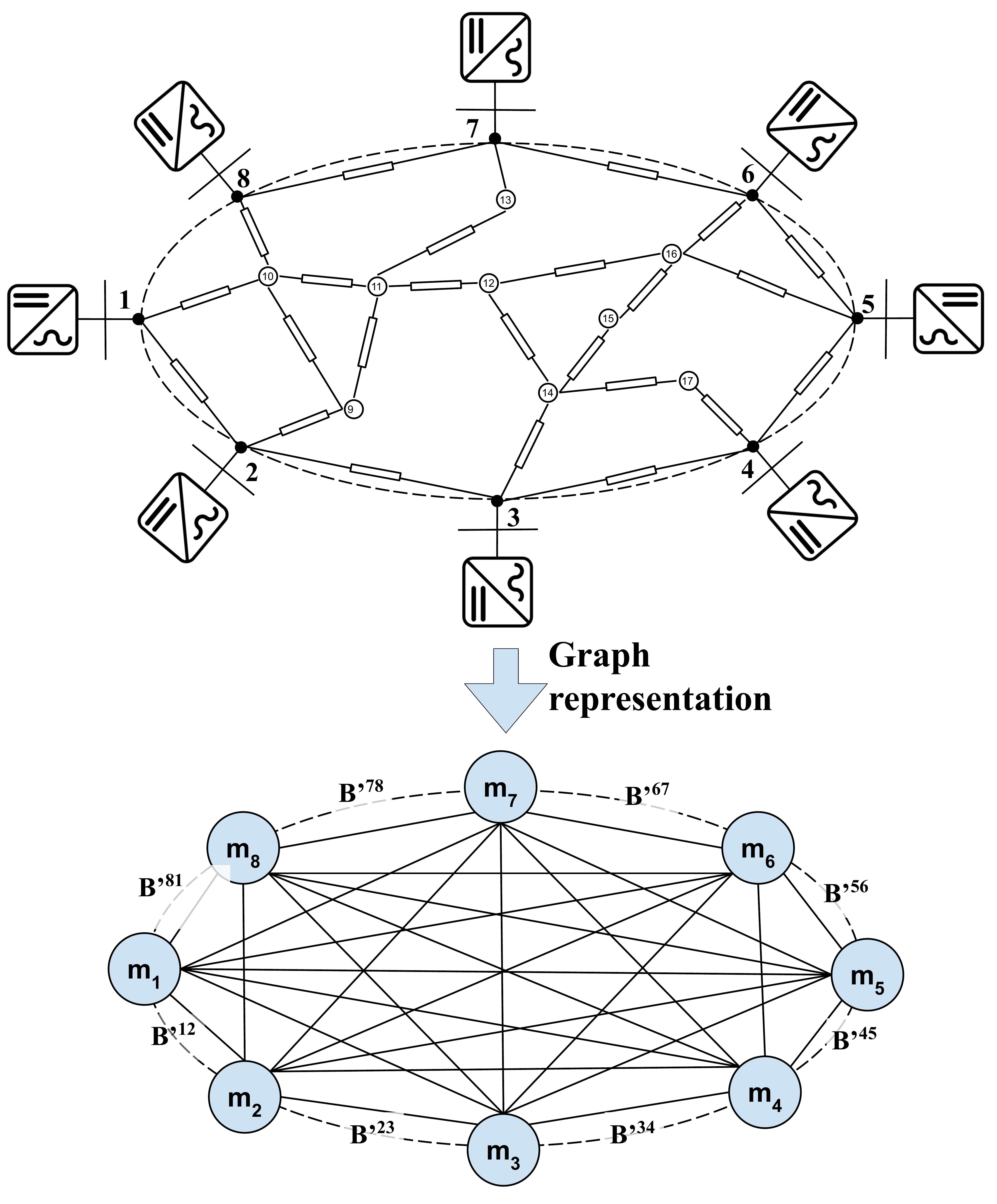}
    \caption{Representation of an inverter-based microgrid by a graph with $m_i$ as node weights and $\mathcal{B}^e$ as branch weights (labels for some lines are omitted).}
    \label{fig:graph_representation}
\end{figure}

Matrix $C$ could be considered as the generalized Laplacian matrix \cite{godsil2001laplacian} for the network graph augmented by node weights equal to droop gains $m_i, \ i=1,\cdots,l$ as shown by Fig. \ref{fig:graph_representation} with real non-negative eigenvalues. Although $C$ is not symmetric, it can be made symmetric by a similarity transform  $M^{-1/2}C M^{1/2} = M^{1/2} \mathcal{B} M^{1/2}$ (since matrix $M$ is diagonal, calculating $M^{1/2}$ is trivial), which proves that the spectrum of $C$ is real. Moreover, both $M^{1/2}$ and $\mathcal{B}$ matrices are positive (semi)definite, so is the product $M^{1/2} \mathcal{B} M^{1/2}$, which proves that the spectrum of $C$ is non-negative. As the matrix $C$ involves dimensionless values of $1/X^{ij}$ (in p.u.) and $m_k$ (in \%), its eigenvalues $\mu_i$ are also dimensionless.  

Having proved that the spectrum of $C$ is non-negative, we can formulate the following theorem, which establishes important stability properties of inverter-based microgrids:
\begin{theorem}\label{theorem_addition_line}
    The addition of any new line, as well as the increase of any existing line susceptance $\mathcal{B}^e = \frac{1}{X^e}$, or the increase of any inverter's droop gain $m_k, k = 1, \cdots, v$ in the system could only increase eigenvalues $\mu_i, \ i=1,\cdots, v$.
\end{theorem}
\begin{proof}
 To check this fact for the line addition, let us consider without loss of generality a new line connection between first and second node, $e = (1,2)$. The new $\hat{C} = C + C_e$ is decomposed into sum of the original $C$ for the graph without added line and $C_e$ is the generalized Laplacian of the graph on $v$ vertices consisting of just the edge $e = (1,2)$ with $X^{1,2}$,

\begin{equation} \label{eq:Laplacian_edge}
    C_e = \frac{1}{X^{1,2}}\begin{pmatrix}m_1 & -m_1 & 0 & \cdots & 0 \\ -m_2 & m_2 & 0 & \cdots & 0 \\ 0 & 0 & 0 & \cdots & 0  \\ \vdots & \vdots  & \vdots  & \ddots & \vdots \\ 0 & 0 & 0 & \cdots & 0
    \end{pmatrix} \ .
\end{equation}
Then one can use Weyl's inequality for eigenvalues of the sum $\hat{C}$ of matrices $C$ and $C_e$ \cite[Theorem~1.3]{so1994commutativity}:
\begin{equation}
    \mu_i \leq \hat{\mu_i} \leq \mu_i + \frac{m_1 + m_2}{X^{1,2}} \ , i=1,\cdots,v \ ,
\end{equation}
where $\hat{\mu_i}$ are eignevalues for $\hat{C}$. The given argument could be applied to any line $e = (i,j)$ addition.
Also, the change of line $e = (i,j)$ susceptance by $\Delta X^e$ gives analogous to \eqref{eq:Laplacian_edge} $C_e$. Hence, the argument extends to the change of the existing line parameters.  

Now, we show the $\mu_i$ increase with the droop gains, $m_k$, increase. Let us represent the increase of $m_k$ by its multiplication by some value $d_k>1$. In this case, we relate original $C$ with $\tilde{C}$ for the increased droop $d_k m_k$ as $\tilde{C} = D C$, where $D = \text{diag}(1, \cdots, d_k, \cdots, 1)$. Next, we use a multiplicative version of the Weyl's inequality \cite[Theorem~4.1]{so1994commutativity}:
    \begin{equation}
        \mu_i \leq \tilde{\mu}_i \leq  d_k\mu_i, \ i=1,\cdots, v \ ,
    \end{equation}
    where $\tilde{\mu}_i$ are eigenvalues of $\tilde{C}$.
\end{proof}

Therefore, Theorem \ref{theorem_addition_line} suggests that the addition of a new line or increase of the susceptance of any existing line makes a microgrid less stable. The same is also true for an increase in any droop gain. The property is entirely consistent with the previous results in \cite{Vorobev2016}, \cite{vorobev2018towards}. This property is a distinctive feature of microgrids in comparison with conventional power systems.

\subsection{Sensitivity Analysis}

Here, the incremental effect of a change in the length $l^{ij}$ of the line $(ij)$ or the droop gains $m_i$ of inverter $i$ on the values of $\mu_i$ is investigated. First, let us consider the sensitivity of the given $\mu_k$ (further in this section, the subscript $k$ is omitted, $\mu = \mu_k$) to the line length $l^{ij}$ using the participation factor analysis as follows:

\begin{eqnarray} \label{eq:sensitivity_length}
        \frac{\partial \mu }{\partial l^{ij}} = \frac{\pmb{\phi}^T \frac{\partial C}{\partial l^{ij}} \pmb{\psi}}{\pmb{\phi}^T \pmb{\psi}} =
         - \frac{([\psi ]_i - [\psi ]_j)^2}{x^{ij} (l^{ij})^2 \sum_{k=1}^v \frac{[\psi]_k^2}{m_k}} \ ,
\end{eqnarray}
where $\pmb{\phi} $ is the left eigenvector of $C$, and $x^{ij}$ is the per-unit length reactance. To obtain the final expression in \eqref{eq:sensitivity_length}, we notice that $\pmb{\phi}  = M^{-1} \pmb{\psi} $. One observes that the sensitivity \eqref{eq:sensitivity_length} is proportional to the eigenvector elements $[\psi ]_i, [\psi ]_j$. Therefore, smaller values of $[\psi ]_i, [\psi ]_j$ will have smaller influence on $\mu $. That says line $l^{ij}$ with insignificant $[\psi ]_i, [\psi ]_j$ should not be included in the cluster, and this is why the suggested identification procedure in Fig. \ref{fig:critical_cluster_flowchart} is based on the elements of eigenvectors $\pmb{\psi}$. 

\begin{figure}
    \centering
    \includegraphics[width=0.4\textwidth]{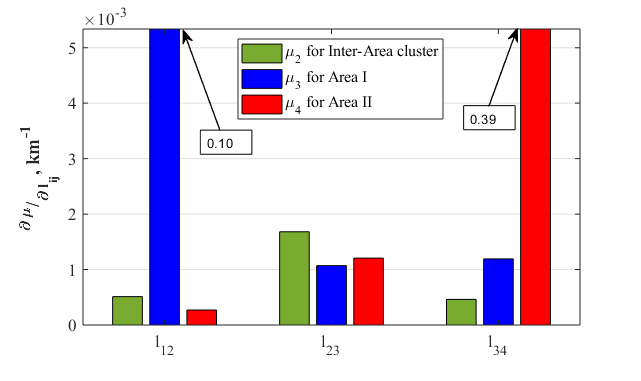}
    \caption{Sensitivities of $\mu$ to the changes in the length of each line $l^{ij}$ according to \eqref{eq:sensitivity_length}. The sensitivities $\frac{\partial \mu_3}{\partial l_{12}} = 0.1$ and $\frac{\partial \mu_4}{\partial l_{34}}=0.39$ are two orders of magnitudes greater than others and exceed the y-axis range. For clarity, the exact values are included.}
    \label{fig:sensit_lengths}
\end{figure}

Secondly, consider the sensitivity to the changes in droop gains $m_k$ that is obtained similarly to \eqref{eq:sensitivity_length}, as follows,

\begin{equation} \label{eq:sensitivity_droop}
    \frac{\partial \mu }{\partial m_k} = \frac{[\psi ]_k}{\sum_{i=1}^v \frac{[\psi]_i^2}{m_i}} \sum_{j: j\neq k} \frac{1}{X^{kj}} ([\psi ]_k - [\psi ]_j) \ .
\end{equation}

The sensitivity to $m_k$ is proportional to $[\psi ]_k$. Therefore, to have a significant response of $\mu $ to changes in $m_k$, the corresponding value of $[\psi ]_k$ should be significant. In other words, inverter $k$ should be a part of the cluster, associated with $\mu$. 

We note, though, that significant values for eigenvector components $[\psi ]_i$ and $[\psi ]_j$ in \eqref{eq:sensitivity_length} do not always lead to high sensitivity of the corresponding $\mu$ to $l_{ij}$, due to $[\psi ]_i - [\psi ]_j$ in the right-hand part of \eqref{eq:sensitivity_length}. Such a case, for example, corresponds to the sensitivity of the $\mu$ eigenvalue for interarea mode, where all the eigenvector elements are considerable (green on Fig.\ref{fig:mode_shapes}), but the sensitivity of $\mu$ to $l_{12}$ and $l_{34}$ is small (Fig.\ref{fig:sensit_lengths}). 


\section{Numerical Validation}\label{sec:numerical}

In this section, numerical validations of the proposed procedure are presented using the $4$-bus, two-area system shown in Fig. \ref{fig:kundur_mmg} (Kundur system). In addition, system parameters are given in Table \ref{tab:parameters}, and they serve as the base-case. Following the procedure for critical cluster identification in Fig. \ref{fig:critical_cluster_flowchart}, $\mu_{cr}=1.97$ for the chosen values of $\rho=1.4$ and $k=3.0$. Next, eigenvalues a of the weighted Laplacian matrix $C$  are calculated as $\mu_1=0$, $\mu_2=0.93$, $\mu_3=1.05$, and $\mu_4=2.06$ (sorted in the increasing order). According to our results, the system is unstable because $\mu_4 > \mu_{cr}$, and by exploring the corresponding eigenvector $\pmb{\phi}_4=[-0.02, 0.06, -0.73, 0.69]^T$, one concludes, that inverters $3$ and $4$ bring the most contribution to the unstable cluster. To find the parameters that have the most effect on the corresponding eigenvalue, formulas \eqref{eq:sensitivity_length} and \eqref{eq:sensitivity_length}  are applied, giving out the droop gains of inverters $3$ and $4$, $m_3$ and $m_4$, and the length of the line $3-4$, $l_{34}$. Further, we vary different system parameters and directly verify the corresponding changes in eigenvalues $\mu$, with simultaneous direct dynamic modelling of the system to show its stability/instability.

\begin{figure}
    \centering
    \includegraphics[width=0.4\textwidth]{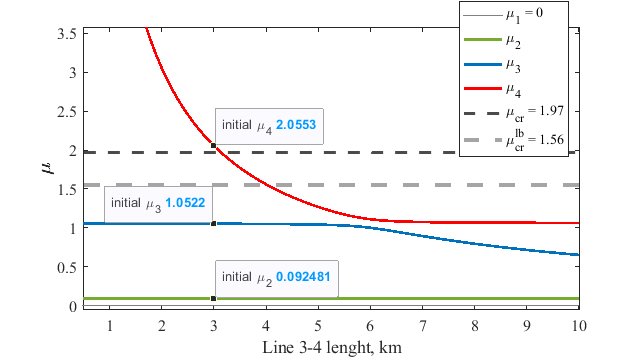}
    \caption{Spectrum $\mu_i$ versus the length of line $3-4$, $l_{34}$, in Area II. 
    The increase in $l_{34}$ stabilizes the system.}
    \label{fig:mu_l34}
\end{figure}

\begin{figure}
    \centering
    \includegraphics[width=0.4\textwidth]{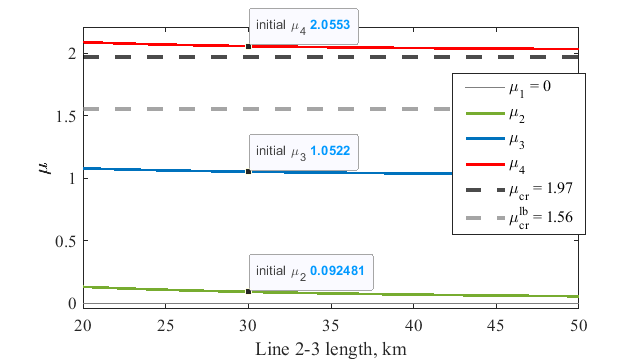}
    \caption{Spectrum $\mu_i$ versus the length of line $2-3$, $l_{23}$, interconnecting two areas. 
    Notice $\mu_4$ remains above the threshold $\mu_{cr}$, that says the system remains unstable even with very long line $l_{23} = 50~km$.}
    \label{fig:mu_l23}
\end{figure}

Fig. \ref{fig:mu_l34} and \ref{fig:mu_l23} show the influence of $l_{34}$ and $l_{23}$ respectively on the eigenvalues of the system matrix $C$ (both line lengths were varied around their initial values). As expected, the biggest eigenvalue is highly sensitive to the variation of $l_{34}$ and almost insensitive to the variation in $l_{23}$. This means, that it is impossible to stabilize the system by varying the length of the line $2-3$ from its base value. 
Variation of $l_{34}$ on the contrary, can stabilise the system, and according to Fig. \ref{fig:mu_l34} the value $l_{34}=3.3~km$ or greater corresponds to stable system. Time-domain simulations shown in Fig. \ref{fig:step_response_all_cases} also justify this observation -- the system becomes stable when subjected to a load change for $l_{34} = 4~km$ and remains unstable with $l_{23} = 50~km$.

Note that changing the line lengths to ensure stabilization can assist in grid planning, but such an approach is not a practical solution for system operation. Therefore, a more feasible way to restore stabilization is through varying the droop gains. Fig. \ref{fig:m1} shows the dependence of eigenvalues of matrix $C$ on $m_1$ the droop gain of inverter $1$. We see that even for minimal values of $m_1$, one of the eigenvalues is above the critical value, hence the system is unstable. When $m_1$ becomes larger, the second eigenvalue crosses the critical level, and the second unstable cluster appears (inverters $1$ and $2$). Anyway, it is impossible to stabilize the initial system by changing $m_1$ only. Fig. \ref{fig:m3} shows the dependence of eigenvalues $\mu$ on the droop gain of inverter $3$ - $m_3$. We see that whenever $m_3 <~2.7\%$, the system is stable and unstable otherwise. Therefore, the initial system can be stabilized by reducing the droop gain $m_3$. Again, the step response in Fig. \ref{fig:step_response_all_cases} illustrates the stability restoration with a smaller $m_3 = 1\%$ and the lack of stability with smaller $m_1 = 1\%$.

\begin{figure}
    \centering
    \includegraphics[width=0.4\textwidth]{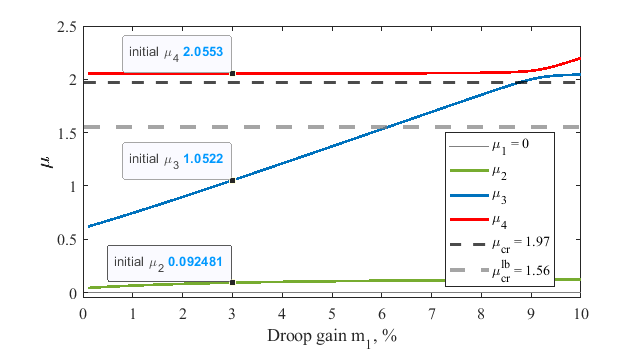}
    \caption{Spectrum $\mu_i$ versus $m_1$ - frequency droop gain of inverter $1$. Eigenvalue, corresponding to Area II cluster ($\mu_4$), remains above the threshold $\mu_{cr}$ even for small $m_1 < 1\%$. Moreover, both Area I and Area II become unstable when $m_1 > 9\%$.}
    \label{fig:m1}
\end{figure}

\begin{figure}
    \centering
    \includegraphics[width=0.4\textwidth]{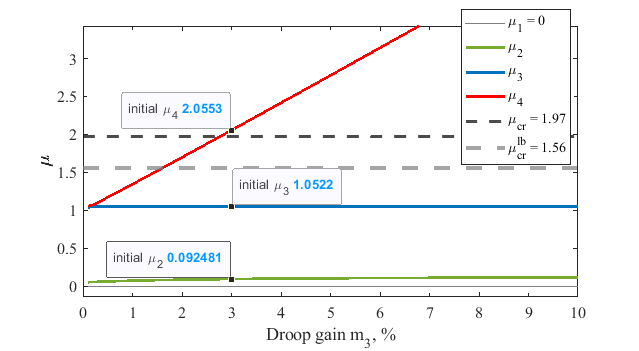}
    \caption{Spectrum $\mu_i$ versus $m_3$ - frequency droop gain of inverter $1$. The system stabilizes with the $m_3$ decrease.}
    \label{fig:m3}
\end{figure}

\begin{figure}
    \centering
    \includegraphics[width=0.4\textwidth]{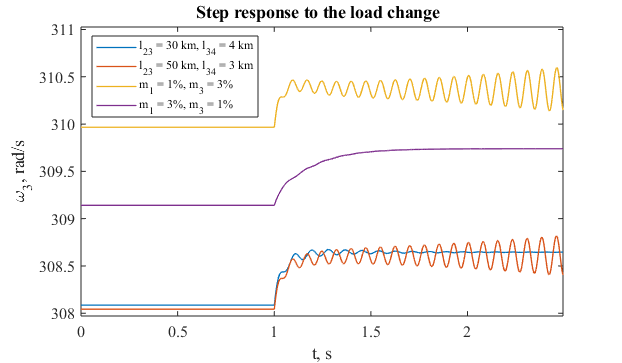}
    \caption{Step response of $\omega_3$ subjected to a 10\% load decrease under different system parameters. The two-area system remains unstable with a large $l_{23} = 50~km$ and a reduced $m_1 = 1\%$, but it stabilizes for a higher $l_{34} = 4~km$ and a reduced $m_3 = 1\%$.}
    \label{fig:step_response_all_cases}
\end{figure}

\section{Conclusion}
In this paper, a novel technique for identification and stabilization of critical clusters in inverter-based microgrids was proposed. The main idea is to decompose a system into a set of clusters by exploiting the properties of the spectrum of the weighted network admittance matrix $C$. This method also allows sorting clusters according to their distances to the stability boundary. The derived theory established a formal basis for the concept of \emph{critical cluster} that is currently missing in the literature.  

We showed that for critical clusters, stability is primarily affected by the length of lines (or impedance) and droop gains. The former provides useful insights for grid planning studies, while the latter benefits system operation. We explicitly demonstrated that to stabilize a microgrid, one should mostly target parameters of the critical clusters - droop coefficients of the corresponding inverters and/or lines connecting them. We also showed that tuning inverters associated with non-critical clusters are futile in regaining network stability. Thus, the paper established the technical principles to explain the oscillatory interactions among inverters and designed an effective procedure to enhance/restore system stability most efficiently.

\begin{table}[ht]
    \centering
    \caption{Parameters of the two-area (four-inverter) system}
    \begin{tabular}{|c|c|c|}
        \hline
        Parameter & Description & Value  \\
        \hline
        $\omega_0$ & Nominal Frequency & $2\pi \cdot 50 \ rad/s$  \\
        $U_b$ & Base Voltage & 230 V\\
        $S_b$ & Base Inverter Rating & 10\,kVA\\
        $\rho $ & R/X ratio & 1.4 \\ 
        $k$ & Droop gain ratio & 3\\
        $R$ &  Line Resistance & $222.2 \ m\Omega/km$\\
        $L$ & Line Inductance & $0.51 \ mH/km$\\
        $m_i$ & Frequency Droop Gain & 3\% \\
        $n_i$ & Voltage Droop Gain & 1\%\\
        $l_{12}$ & Line $1-2$ length & 6 $km$ \\
        $l_{23}$ & Line $2-3$ length & 30 $km$ \\
        $l_{34}$ & Line $3-4$ length & 3 $km$ \\
        $Z_1$ & Bus 1 load & $(20+1j) \Omega$\\
        $Z_2$ & Bus 2 load & $(25+1j) \Omega$\\
        $Z_3$ & Bus 3 load & $(20+4.75j) \Omega$\\
        $Z_4$ & Bus 4 load & $(40+12.58j) \Omega$\\
        \hline
    \end{tabular}
    
    \label{tab:parameters}
\end{table}

\ifCLASSOPTIONcaptionsoff
  \newpage
\fi

\bibliographystyle{IEEEtran}
\bibliography{bibl}
\end{document}